\documentclass[12pt,a4paper]{article}
\usepackage{amssymb,latexsym,amsmath,amsthm,graphicx}
\usepackage{float}
\usepackage{amsthm,amssymb,amsmath,graphicx}
\usepackage{amsfonts,xcolor,complexity}
\usepackage{cancel,enumitem}
\usepackage{algpseudocode}
\usepackage{fullpage}
\usepackage[margin=0.98in]{geometry}
\usepackage[linesnumbered,ruled]{algorithm2e}
\usepackage{mathrsfs} 
\usepackage{hyperref}
\hypersetup{
colorlinks=true, %set true if you want colored links
   linkcolor=blue,  %choose some color if you want links to stand out
}
\DeclareGraphicsExtensions{.svg,.png,.jpg,.pdf}
\newtheorem{theorem}{Theorem}
\newtheorem{defn}{Definition}
\newtheorem{remark}{Remark}
\newtheorem{claim}{Claim}[theorem]

\newtheorem{corollary}{Corollary}

\newtheorem{obs}{Observation}
\newtheorem{lemma}{Lemma}
\usepackage{epsfig}

\theoremstyle{plain}
\newcommand{\lab}{{\sf label}}
\newcommand{\lin}{\ensuremath{\equiv_{\sf lin}}}

\newcommand{\aff}{\sf aff}

\newcommand{\vd}{{\sf VD}}
\newcommand{\var}{{\sf var}}

\newcommand{\proj}{{\sf proj}}
\newcommand{\homo}{{\sf homo}}

\newcommand{\ed}{{\sf RED_S}}

\newcommand{\pit}{{\sf PIT}}
\newcommand{\vdequiv}{{\sf VD-EQUIV }}

\title{Linear Projections of the Vandermonde Polynomial}
\author{ C. Ramya   \and
  B. V. Raghavendra Rao}
  \date{  {Department of Computer Science and Engineering} \\
IIT Madras, Chennai, INDIA \\
\texttt{ \{ramya,bvrr\}@cse.iitm.ac.in}\\
}

\begin{document}

\maketitle

\begin{abstract}
An $n$-variate Vandermonde polynomial is the determinant of the $n\times n$ matrix where the $i$th column is the vector $(1, x_i, x_i^2, \dots, x_i^{n-1})^T$. Vandermonde  polynomials play a crucial role in the theory of alternating polynomials and    occur in Lagrangian   polynomial interpolation as well as in  the theory of   error correcting codes.   In this work we study structural and computational aspects of  linear projections of  Vandermonde polynomials.   

Firstly, we consider the problem of testing if a given polynomial is  linearly equivalent to the Vandermonde polynomial. We obtain a deterministic polynomial time algorithm   to test if $f$ is linearly equivalent to the Vandermonde polynomial  when $f$ is given as product of linear factors. In the case when $f$ is given as a black-box  our algorithm runs in   randomized polynomial time. Exploring the structure of projections of Vandermonde polynomials further,   we  describe the group of symmetries of a Vandermonde polynomial and show that the associated Lie algebra is simple. 
Finally, we study arithmetic circuits built over projections of  Vandermonde polynomials. We show universality property for some of the models and obtain  a lower bounds against sum of projections of Vandermonde determinant. 

\end{abstract}

\section{Introduction}
\label{sec:intro}
The  $n\times n$ symbolic Vandermonde matrix  is given by
\begin{equation}
\label{eq:eq1}
V =  \left[ {\begin{array}{ccccc} 
 1 & 1 & 1 & 1 & 1 \\
 x_1 & x_2 & \cdots & \cdots & x_n\\
x_1^2 & x_2^2 & \cdots &\cdots & x_n^2\\
\vdots & \vdots & \vdots & \vdots & \vdots \\
\vdots & \vdots & \vdots & \vdots & \vdots \\
x_1^{n-1} & x_2^{n-1} & \cdots &\cdots & x_n^{n-1}\\
\end{array} } \right]
\end{equation}
where $x_1,\ldots, x_n$ are variables.  The determinant of the symbolic  Vandermonde matrix is a homogeneous polynomial of degree $\binom{n}{2}$ given by ${\sf \vd_n(x_1,\ldots,x_n)} \triangleq {\sf det}( V) = \prod_{i<j}(x_i-x_j)$ and is known as the $n$-variate {\em Vandermonde polynomial}. An {\em alternating polynomial} is one that changes sign when any two variables of $\{x_1,\ldots,x_n\}$ are swapped.  Vandermonde polynomials are central to the theory of alternating polynomials. In fact, any alternating polynomial is divisible by the Vandermonde polynomial~\cite{Mat,MD98}.  Further, Vandermonde matrix and Vandermonde polynomial occur very often in the theory of error correcting codes and are useful in Lagrangian interpolation. 

Linear projections are the most important form of reductions in Algebraic Complexity Theory developed by Valiant~\cite{Val79}. Comparison between classes of polynomials in Valiant's theory depends on the types of linear projections. (See \cite{Bur} for a detailed exposition.)  Taking a geometric view on linear projections of polynomials, Mulmuley and Shohoni~\cite{MS01} proposed the study of geometry of orbits of polynomials under the action of ${\sf GL}(n,\mathbb{F})$, i.e, the group of $n\times n$ non-singular matrices over $\mathbb{F}$.  This lead to the development of Geometric Complexity Theory, whose primary objective is to classify families of polynomials based on the geometric and representation theoretic structure of their  ${\sf GL}(n,\mathbb{F})$ orbits. 

In this article, we investigate   computational and  structural  aspects of linear projections of the family  $\vd = (\vd_n)_{n\ge 0}$ of Vandermonde polynomials over the fields of real and complex numbers. 
%More specifically, we consider the computational problem of testing if a given polynomial is linearly equivalent to the Vandermonde polynomial and 
 
%Two fundamental computational problems in Algebraic Complexity Theory are - Polynomial Identity Testing ({\sf PIT} for short) and Polynomial Equivalence ({\sf POLY-EQUIV} for short). In the {\sf PIT} problem we ask if the given polynomial $f$ is identically zero or not.
Firstly, we consider the polynomial equivalence problem when one of the polynomials is fixed to be the  Vandermonde polynomial. Recall that, in the  polynomial equivalence problem ({\sf POLY-EQUIV})  given a pair of polynomials $f$ and $g$ we ask if $f$ is equivalent to $g$ under a non-singular linear change of variables, i.e., is there a $A\in  {\sf GL}(n,\mathbb{F})$ such that $f(AX) = g(X)$, where $X= (x_1,\ldots, x_n)$? {\sf POLY{-}EQUIV} is one of the fundamental computational problems over polynomials and received significant attention in the literature. 
%\textcolor{red}{
%{\sf POLY-EQUIV} can be solved in {\sf PSPACE} over algebraically closed fields~\cite{AS06} and in doubly exponential time over reals~ \cite{DH88}. However,  it is not known if the problem is even decidable over the rational numbers. Over finite fields, {\sf POLY-EQUIV} is in $\NP\cap coAM$~\cite{Thi98}.}

{\sf POLY-EQUIV} can be solved in {\sf PSPACE} over reals~\cite{JC88} and any  algebraically closed field~\cite{Sax06}, and is  in $\NP\cap {\sf co}\mbox{-}{\sf AM}$~\cite{Thi98} over finite fields.  However,  it is not known if the problem is even decidable over the field of  rational numbers~\cite{Sax06}.    Saxena~\cite{Sax06} showed that ${\sf POLY\text{-}EQUIV}$  is at least as hard as the graph isomorphism problem even in the case of   degree three forms.  Given the lack of progress on the general problem, authors have focussed on special cases over the recent years. Kayal~\cite{Kayal11} showed that testing if a given polynomial $f$ is linearly equivalent to the elementary symmetric polynomial, or to the power symmetric polynomial can be done in randomized polynomial time. Further, in~\cite{Kayal12}, Kayal obtained randomized polynomial time  algorithms for {\sf POLY-EQUIV} when one of the polynomials is either the determinant or permanent and the other polynomial is given as a black-box. 
 We consider the problem of testing equivalence to Vandermonde polynomials: 

\begin{center}
\begin{minipage}{11.95 cm}
\textbf{Problem : }{\sf VD-EQUIV}\\
\textbf{Input :} $f\in\mathbb{F}[x_1,\ldots,x_n]$\\
\textbf{Output :} Homogeneous linearly independent linear forms $L_1,L_2,\ldots,L_n$ such that $f=\vd( L_1,L_2,\ldots,L_n )$  if they exist, else output \texttt{`No such equivalence exists'}.
\end{minipage}
\end{center} 
 
\begin{remark}
Although Vandermonde polynomial is a special form of determinant, randomized polynomial time algorithm to test equivalence to determinant polynomial due to \cite{Kayal12} does not directly give an algorithm for {\sf VD-EQUIV}.  
\end{remark}
  
We show that {\sf VD-EQUIV} can be solved in deterministic polynomial time when $f$ is given as a product of linear factors (Theorem~\ref{thm:vdequiv}). Combining this with Kaltofen's factorization algorithm, \cite{Kal89},  we get a randomized polynomial time algorithm for {\sf VD-EQUIV} when $f$ is given as a black-box.

%From \cite{IK04}, we know that it is hard to derandomize {\sf PIT} with proving strong arithmetic circuit lower bounds. Hence 

 For an $n$-variate polynomial $f\in \mathbb{F}[x_1,\ldots, x_n]$, the group of symmetry $\mathscr{G}_f$ of $f$ is the set of non-singular matrices that fix the polynomial $f$. The group of symmetry of a polynomial  and the associated Lie algebra  have significant importance in geometric complexity theory.  More recently, Kayal~\cite{Kayal12} used the structure of Lie algebras of permanent and determinant in his algorithms for special cases of {\sf POLY-EQUIV}. Further, Grochow~\cite{Gro12} studied the problem of testing conjugacy of matrix Lie algebras. In general, obtaining a complete description of group of symmetry and the associated Lie algebra of a given family of polynomials is an interesting task. 

%{In \cite{Kayal12} and subsequently in \cite{Gro12} the problem of {\sf POLY-EQUIV} was handled using the concept of Lie Algebras. Lie Algebra is central to the Mulmuley-Sohoni Geometric Complexity Theory Program \cite{MS01}. A matrix Lie Algebra is a set of $n\times n$ matrices that are closed under scalar multiplication, matrix addition and the operation $[A,B]= AB-BA$. More abstractly, for a polynomial $f$, the lie algebra of $f$ is precisely the subspace of $\mathbb{F}^{n\times n}$ tangent to the group of symmetries of $f$ at the identity matrix \cite{Kayal12}. For testing equivalence of a polynomial $f$ to determinant(resp. permanent) polynomial, Kayal \cite{Kayal12} used a concept closely related to Lie Algebras of $f$ and the determinant( resp. permanent) called {\em Lie Algebra Conjugacy}.}
 %{In particular Kayal\cite{Kayal12} showed that if $f$ is a linear projection of $g$ then their corresponding lie algebras are conjugates. (See section \ref{sec:lie} for details) The fact that the converse of the above statement is not true and that given black-box access to a polynomial $f$ constructing a basis for lie algebra of $f$ is equivalent to {\sf PIT} stand as a major obstruction  to solving the Polynomial Equivalence problem using Lie Algebras. } 
 In this paper we obtain a description of the group of symmetry for Vandermonde polynomials (Theorem~\ref{thm:group-of-sym}). Further, we show that the associated Lie algebra for Vandermonde polynomials is simple (Lemma~\ref{lem:lie-vand}).

Finally, we explore linear projections of Vandermonde polynomials as a computational model.  We prove closure properties (or lack of) and lower bounds for representing a polynomial as sum of projections of  Vandermonde polynomials (Section~\ref{sec:models}).
 
%It is important to note that not all polynomials can be expressed as linear projections of Vandermonde polynomial. In fact, no irreducible polynomial is expressible as a linear projection of Vandermonde polynomial.Interestingly we observe that every polynomial $f\in\mathbb{F}[x_1,\ldots,x_n]$ can be written as sum of linear projections of Vandermonde polynomials. In this regard, we define arithmetic circuit classes built over Vandermonde polynomials. This is in spirit similar to the symmetric model defined by Shpilka \cite{Shp02} - a depth-2 circuit with a symmetric gate at the top and
%plus gates at the bottom computing a symmetric function in linear polynomials.

%
%\begin{theorem}
%\begin{itemize}
%\item[(i)] Let $R$ be the univariate polynomial ring over the field $\mathbb{F}$. Then $R \subseteq \Sigma\cdot \vd_{proj}$.
%\item[(ii)] The classes $\Sigma\cdot\vd_{homo}$ and  $\Sigma\cdot\vd_{aff}$ are universal. Any polynomial $f\in\mathbb{F}[x_1,\ldots,x_n]$ can be expressed in  $\Sigma\cdot\vd_{homo}$ and  $\Sigma\cdot\vd_{aff}$.
%\end{itemize}
%\end{theorem}
%
%Finally we investigate the limitations of circuits built over Vandermonde polynomials to 
%the existing circuit classes :
%

%

\section{Preliminaries} 
\label{sec:prelim}
Throughout the paper, unless otherwise stated, $\mathbb{F}\in \{\mathbb{C}, \mathbb{R}\}$.  
We briefly review different types of projections of polynomials that are useful for the article. For a more detailed exposition, see~\cite{Bur}.

\begin{defn}{\em $($Projections$)$.}
Let $f,g\in \mathbb{F}[x_1,x_2,\ldots,x_n]$.   We say that $f$ is projection reducible to $g$ denoted by $f \le g$, if there are linear forms $\ell_1,\ldots, \ell_n \in \mathbb{F}[x_1,\ldots, x_n]$ such that $f = g(\ell_1, \ldots, \ell_n)$. Further, we say 
 \begin{itemize}
\item $f \le_{\proj} g$ if $\ell_1,\ldots, \ell_n \in \mathbb{F}\cup \{x_1,\ldots, x_n\}.$
\item $f \le_{\homo} g $ if $\ell_1,\ldots, \ell_n$ are homogeneous linear forms.
\item $f \le_{\aff} g$ if $\ell_1,\ldots, \ell_n$ are affine linear forms.
 \end{itemize}
\end{defn}

Based on the types of projections, we consider the following classes of polynomials that are projections of the Vandermonde polynomial.
\begin{align*}
\vd &= \{\vd(x_1,x_2,\ldots,x_n)\mid n\geq 1\}; \text{ and } \\
\vd_{\proj} &= \{ \vd(\rho_1,\rho_2,\ldots,\rho_n) \mid \rho_i \in (X\cup \mathbb{F}), \forall i\in[n] \}; \text{ and } \\
\vd_{\homo} &= \{\vd(\ell_1,\ell_2,\ldots,\ell_n) \mid \ell_i \in \mathbb{F}[x_1,x_2,\ldots,x_n],~  \deg(\ell_i)=1,~ \ell_i(0)=0~~ \forall i\in[n]  \};\text{ and }\\
\vd_{\aff} &= \{\vd(\ell_1,\ell_2,\ldots,\ell_n) \mid \ell_i \in \mathbb{F}[x_1,x_2,\ldots,x_n],~  \deg(\ell_i)\leq 1~~ \forall i\in[n]  \};
\end{align*}
 Among the different types mentioned above, the case when $\ell_1,\ldots, \ell_n$ are homogeneous and linearly independent is particularly interesting. 
Let $f,g\in \mathbb{F}[x_1,\ldots, x_n]$. $f$ is said to be linearly equivalent to $g$ (denoted by $f \lin g$) if  $g \le_{\homo} f$ via a set of linearly independent homogeneous linear forms $\ell_1,\ldots, \ell_n$. In the language of invariant theory, $f \lin g$ if and only if $g$ is in the ${\sf GL}(n,\mathbb{F})$ orbit of $f$.

 The group of symmetry of a polynomial is one of the fundamental objects associated with a polynomial: 
 \begin{defn}
 	Let $f\in\mathbb{F}[x_1,x_2,\ldots,x_n]$. The group of symmetries of $f$ ( denoted by  $\mathscr{G}_f$) is  defined as: 
    $$ \mathscr{G}_f  = \{A~|~ A \in {\sf GL}(n, \mathbb{F}), f(A {\sf x}) = f({\sf x})\}. $$ 	
  i.e., 	the group  of invertible $n\times n$ matrices $A$ such that $f(A \mathbf{x})=A(\mathbf{x})$. 
 \end{defn}

The  Lie algebra of a polynomial $f$ is the tangent space of $\mathscr{G}_f$ at the identity matrix and is defined as follows: 
 \begin{defn}[\cite{Kayal12}]
 	Let $f\in\mathbb{F}[x_1,x_2,\ldots,x_n]$. Let $\epsilon$ be a formal variable with $\epsilon^2=0$. Then $\mathfrak{g}_f$ is defined to be the set of all matrices $A\in\mathbb{F}^{n\times n}$ such that \begin{center}
 		$f((\mathbf{1}_n+\epsilon A)\mathbf{x})=f(\mathbf{x})$.
 	\end{center}
 \end{defn}

It can be noted that $\mathfrak{g}_f$ is non-trivial only when $\mathscr{G}_f$ is a continuous group. For a random polynomial, both $\mathscr{G}_f$ as well as $\mathfrak{g}_f$  are trivial.  

For a polynomial $f\in \mathbb{F}[x_1,\ldots,x_n]$, $k\ge 0$, let $\partial^{=k}( f)$ denote the $\mathbb{F}$-linear span of the set of all partial derivatives of $f$ of order $k$, i.e.,
$$\partial^{=k} (f) \triangleq \mathbb{F}\mbox{-}{\sf Span }\left\{  \frac{ \partial^k f}{\partial x_{i_1}\dots \partial x_{i_k}}~\vline~ i_1,\ldots, i_k \in [n] \right\}. $$

\section{Testing Equivalence to Vandermonde Polynomials}

%For any $f\in\mathbb{F}[x_1,x_2,\ldots,x_n]$ if  $f\in \vd_{homo}$, then $f$ is a product of linear factors.  However, by Lemma $\ref{lem:vd-univ}$, the converse is not necessarily true. Thus, it is natural to ask if a given polynomial $f\in\mathbb{F}[x_1,\ldots,x_n]$ is  linearly equivalent to Vandermonde polynomial. More formally, in this section we consider the following computational problem :
%
%\begin{center}
%\begin{minipage}{11.95 cm}
%\textbf{Problem : }{\vdequiv}\\
%\textbf{Input :} $f\in\mathbb{F}[x_1,\ldots,x_n]$\\
%\textbf{Output :} Homogeneous linear forms $L_1,L_2,\ldots,L_n$ such that $f=\vd( L_1,L_2,\ldots,L_n )$  if extist, else output \texttt{`No such projection exists'}.
%\end{minipage}
%\end{center}  

Recall the problem $\vdequiv$ from Section \ref{sec:intro}. In this section, we obtain an  efficient algorithm for $\vdequiv$. The complexity of the algorithm depends on the input representation of the polynomials.  When the polynomial is given as product of linear forms, we show:
\begin{theorem}
\label{thm:vdequiv}
There is a deterministic polynomial time algorithm for  $\vdequiv$   when the input polynomial  $f$ is given as a  product of  homogeneous linear forms.
\end{theorem}
 \noindent The proof of Theorem~\ref{thm:vdequiv} is based on the correctness of Algorithm~\ref{algo-vdproj} described next. 
 \vspace*{4 mm}

%\subsection{Proof of Theorem \ref{thm:vdequiv}(i)}
{ 
    \begin{algorithm}[H]
    \label{algo-vdproj}
    \SetKwInOut{Input}{Input}
    \SetKwInOut{Output}{Output}
    \Input{$f=\ell_1\cdot\ell_2\cdots \ell_p$}
    \Output{\texttt{`$f$ is linearly equivalent to $\vd(x_1,\ldots,x_n)$'} if $f \equiv_{\sf lin} \vd$.  Else \texttt{`No'}}
\label{algo:check1}    \If{$p\neq\binom{n}{2}$ for any $n<p$ \textbf{or} $\dim(span\{\ell_1,\ell_2,\ldots,\ell_p\})\neq n-1$} 
      {
        return \texttt{`No such equivalence exists'} 
      }\label{algo:check2}
%	\If{$\dim(span\{\ell_1,\ell_2,\ldots,\ell_p\})>n-1$}      
%      {
%    return \texttt{`No such projection exists'} 
%      } \label{algo:check2}
    \Else
    {
	$S\leftarrow \{\ell_1,\ell_2,\ldots,\ell_p\}$ \label{algo:sets}  \\
\label{algo:choice}
Let $T^{(0)}=\{ r_1,r_2,\ldots,r_{n-1}\}$ be  $n-1$ linearly independent linear forms in $S$.    \\
  
    $i \leftarrow 1$ \\
    $S' \leftarrow S\setminus T^{(0)} $\\
    \While{true}
    {
    $D_i\leftarrow \{ a+b , a-b, b-a \mid a,b \in T^{(i-1)}\}$ \Comment{$D_i$ is the set of differences} \label{algo:diff}\\
	$T^{(i)}\leftarrow \{T^{(i-1)} \cup D_i\} \cap S$   \\  
    $S'\leftarrow S\setminus T^{(i)}$ \\
    \If{$|S'|=0$} 
    { Output \texttt{`$f$ is linearly equivalent to $\vd(x_1,\ldots,x_n)$'} \label{algo:accept}\\
	\textbf{Exit.}
    }
    \If{$T^{(i-1)}=T^{(i)}$ \text{and} $|S'|\neq 0$}
    {    
	   Output \texttt{`No such equivalence exists'}\\
	    \textbf{Exit.}
    }    
    }
    }     \caption{$\vdequiv$}
\end{algorithm}
}
%\subsubsection{Correctness of Algorithm \ref{algo-vdproj}}
As a first step, we observe that lines (\ref{algo:check1})-(\ref{algo:check2}) of algorithm  are correct:
\begin{obs} If $f=\ell_1\cdot \dots \cdot \ell_p \lin \vd$ then $p=\binom{n}2$ and $\dim(span\{\ell_1,\ell_2,\ldots,\ell_p\})= n-1$.
\end{obs}
\begin{proof}
Clearly if $f\lin \vd$  we have $p=\binom{n}{2}$. For the second part suppose $f = \vd(L_1,\ldots, L_n)$, for some linearly independent homogeneous linear forms $L_1,\ldots, L_n$. Then $\{\ell_1,\ldots,\ell_p\} = \{L_i - L_j~|~ i<j\}$, and therefore $\dim(span\{\ell_1,\ell_2,\ldots,\ell_p\}) = n-1$. 
\end{proof}
The following theorem proves the correctness of the Algorithm \ref{algo-vdproj}.
\begin{theorem}
$f\lin \vd$  if and only if Algorithm $\ref{algo-vdproj}$ outputs \textnormal{\texttt{`$f$ is linearly equivalent to $\vd(x_1,\ldots,x_n)$'}}.
\end{theorem}
\begin{proof}
We first argue the forward direction. Suppose  there are  $n$ homogeneous linearly independent linear forms $L_1',L_2',\ldots,L_{n}'$ such that $f=\ell_1\cdot\ell_2\cdots \ell_p=\prod\limits_{i<j\\i,j\in[n]}(L_i'-L_j')$. Consider the linear forms $L_1 = L_1'-L_n', ~~L_2 = L_2'-L_n', \cdots , L_{n-1} = L_{n-1}'-L_n'.$
Then \begin{equation}
\label{eqn:vd-homo}
 f=\ell_1\cdot\ell_2\cdots \ell_p=\prod_{i=1}^{n-1}L_i\cdot \prod_{i<j}(L_i-L_j).
\end{equation}
Let $S\triangleq\{\ell_1,\ell_2,\ldots,\ell_p\}$ as in line (\ref{algo:sets}) of the algorithm. By	equation (\ref{eqn:vd-homo}), we have 
$$S=\{L_1,L_2,\ldots,L_{n-1}\} \cup \{ L_i-L_j \mid  i<j,~~ i,j\in[n-1] \}.$$
Let $S_1 \triangleq\{L_1,L_2,\ldots,L_{n-1}\}$ and  $S_2 \triangleq\{ L_i-L_j \mid  i<j,~~ i,j\in[n-1]\}$. Consider the undirected complete graph $G$ with vertices $\{v_1,v_2,\ldots,v_{n-1}\}$. 
For every vertex $v_i\in V(G)$, let $\lab(v_i)$ denote the linear form $L_i$ associated with the vertex $v_i$. Similarly for every edge $e=(v_i,v_j)\in E(G)$ let $\lab(e)$ be defined as follows :

\begin{equation}
\label{eq:label}
\lab(e)  =  \begin{cases} L_i-L_j &\mbox{if $i<j$ } \\
L_j-L_i & \mbox{if $j<i$ }  \end{cases}
\end{equation}

Using notations used in  line (\ref{algo:choice}) of Algorithm~\ref{algo-vdproj}, we have $\{r_1,r_2,\ldots,r_{n-1}\}\subseteq S$. Observe that for every $i\in[n-1]$ the linear form $r_i$ corresponds to either a vertex or an edge label in $G$. Let $Q_1\triangleq \{r_1,r_2,\ldots,r_{n-1}\}\cap S_1$ and  $Q_2\triangleq \{r_1,r_2,\ldots,r_{n-1}\}\cap S_2$. Suppose $|Q_2|=k$ and $|Q_1|=n-k-1$, linear forms in $Q_1$ correspond to labels of vertices in $V(G)$ and linear forms in $Q_2$ correspond to labels of edges in $E(G)$. For some $k\geq 0$, let
\begin{align*}
Q_1 &= \{\lab({u_1}),\lab({u_2}),\ldots,\lab({u_{n-k-1}})\} \hspace*{10 mm} \text{for $u_1,u_2,\ldots,u_{n-k-1}\in V(G)$}\\
Q_2 &= \{\lab({e_1}),\lab({e_2}),\ldots,\lab(e_k)\} \hspace*{19 mm} \text{for $e_1,e_2,\ldots,e_{k}\in E(G)$}
\end{align*}

%$$ for some $k\ge 0$.

\noindent Let $G[r_1,\ldots, r_{n-1}]$ denote the sub-graph $\{u_1,\dots, u_{n-k-1}\} \cup \{e_1,\dots, e_k\}$, i.e., consisting of edges with labels in $Q_2$ and vertices incident on them and vertices with labels in $Q_1$. 

%\textcolor{red}{Let $G[r_1,\ldots, r_{n-1}]$ denote the sub-graph induced on the vertex set $\{u_1,u_2,\ldots,u_{n-k-1}\}$ and the vertices incident on  $e_1,e_2,\ldots,e_{k}$ along with the edges $e_1,e_2,\ldots,e_{k}$.}

We need the following claim:
\begin{claim}
\label{claim:one vertex}
For any choice of linearly independent linear forms $\{r_1,r_2,\ldots,r_{n-1}\}$ by the algorithm in line $(\ref{algo:choice})$, any connected component $C$ in $G[r_1,r_2,\ldots,r_{n-1}]$ has   exactly one vertex  with label in $Q_1$. More formally, if $Q_C\triangleq\left(\bigcup_{v\in Q_1}\lab({v})\right)\bigcap\left( \bigcup_{w\in V(C)}\lab{(w)}\right)$
then $|Q_C| =1$.

\end{claim}
\begin{proof}[Proof of Claim~\ref{claim:one vertex}]
Proof is by contradiction. Suppose there is a connected component $C$ in $G[r_1,r_2,\ldots,r_{n-1}]$ with $|Q_C| \geq 2$. Let $v_i,v_j \in Q_C$. Assume without loss of generality that $i<j$. Consider the path $\bar{P} = (v_i,e_{c_1},e_{c_2},\ldots,e_{c_{|\bar{P}|-1}},v_j)$ between $v_i$ and $v_j$ in the connected component $C$, where $e_{c_1}, \ldots , e_{c_{|\bar{P}|}-1}$ are edges. From the definition of $G$, we know that there are constants $\alpha_1,\ldots \alpha_{|\bar{P}| -1} \in \{-1,1\}$ such that $$\alpha_1\lab{(e_{c_1})}+\alpha_2\lab{(e_{c_2})}+\cdots+\alpha_{|\bar{P}|-1}\lab{(e_{c_{|\bar{P}|-1}})} = \lab{(v_i)}-\lab{(v_j)}.$$  Therefore, $\{\lab{(v_i)},\lab{(e_{c_1})},\lab{(e_{c_2})},\ldots,\lab{(e_{c_{|\bar{P}|-1}})},\lab{(v_j)}\}$ is a linearly dependent set. Since $C$ is a connected component in $G[r_1,r_2,\ldots,r_{n-1}]$ we have that the set of linear forms $\{\lab{(v_i)},\lab{(e_{c_1})},\lab{(e_{c_2})},\ldots,\lab{(e_{c_{|\bar{P}|-1}})},\lab{(v_j)}\} \subseteq \{r_1,r_2,\ldots,r_{n-1}\} $, hence a contradiction.  Now, suppose there exists a connected component $C$ with $Q_C=\emptyset$. Let $v$ be any vertex in $C$. Clearly, $\{\lab{(v)}\cup \{r_1,\ldots, r_{n-1}\}\}$ is a linearly independent set, a contradiction since $\dim (\mathbb{F}\mbox{-}{\sf span}( S)) = n-1$.
\end{proof}
 
Now, the following claim completes the proof of the forward direction:
\begin{claim}
\label{cl:reach}
\begin{itemize}
\item[(i)] If $f\lin \vd $ then there exists an $m$ such that $\{L_1,L_2,\ldots,L_{n-1}\}\subseteq T^{(m)}$. 
\item[(ii)]For any $m$, if $\{L_1,L_2,\ldots,L_{n-1}\}\subseteq T^{(m)}$ then the set $T^{(m+1)} = S$ and the algorithm outputs \textnormal{ \texttt{`$f$ is linearly equivalent to $\vd (x_1,\ldots,x_n)$'}} in line \ref{algo:accept}. 
\end{itemize}
\end{claim}
\begin{proof}[Proof of Claim~\ref{cl:reach}]
(i)    Let $C$ be a connected component in $G[r_1,\ldots,r_{n-1}]$. By Claim \ref{claim:one vertex} we have $|Q_C|=1$. Let $Q_C=\{b\}$ and $L=\lab{(b)}$. We argue by induction on $i$ that for every vertex $v\in V(C)$ with  $dist(L,v)\leq i$, $\lab{(v)}\in T^{(i)}$. Base case is when $i=0$ and follows from the definition of $T^{(0)}$. For the induction step, let $u\in V(C)$ be such that $(u,v)\in E(G[r_1,\ldots , r_{n-1}])$ and $dist(L,u)\leq i-1$. By the induction hypothesis, we have $\lab{(u)}\in T^{(i-1)}$.
  Also, since $\lab(u,v)\in \{r_1,\ldots, r_{n-1}\}=T^{(0)}$, we have   $\lab{(u,v)}\in T^{(i-1)}$.   By line \ref{algo:diff} of the algorithm, the linear form $\lab{(u,v)}+\lab{(u)}\in D_i$ where $D_i$ is the set of differences in the $i^{th}$ iteration of the while loop. Now, by the definition of labels in \ref{eq:label}, $(L_v-L_u)+L_u\in D_i$ if $v<u$ or  $L_u-(L_u-L_v)\in D_i$ if $u<v$. In any case, $L_v=\lab{(v)}\in T^{(i)}$ as required. Now, if $m\ge n-1$ then we have $\{L_1,\ldots, L_{n-1}\} \subseteq T^{(m)}$.

\noindent (ii) If $\{L_1,L_2,\ldots,L_{n-1}\}\subseteq T^{(m)}$ then clearly $T^{(m)} \cup D_m  = S$. Hence $T^{(m+1)}=S$ and algorithm outputs \textnormal{ \texttt{`$f$ is linearly equivalent to $\vd (x_1,\ldots,x_n)$'}} in line \ref{algo:accept}.
\end{proof} 

   Suppose Algorithm~\ref{algo-vdproj} outputs \textnormal{ \texttt{`$f$ is linearly equivalent to $\vd (x_1,\ldots,x_n)$'}} in $k$ steps. Consider the polynomial $\vd(\ell,r_1,r_2,\ldots,r_{n-1})$ where $\{r_1,r_2,\ldots,r_{n-1}\}$ is the linearly independent set chosen in line \ref{algo:choice} of  Algorithm~\ref{algo-vdproj}  and $\ell$ is an arbitrary linear form such that the set  $\{\ell, r_1,r_2,\ldots,r_{n-1}\}$ is linearly independent. Then, we have   $\ell_1\ell_2\cdots\ell_p =  \vd(\ell,\ell -r_1,\ell-r_2,\ldots,\ell - r_{n-1})$. 
   % Then consider the set $|S'|\neq 0$ in the $k^{th}$ step. Hence the algorithm would not have terminated which is a contradiction.
\end{proof}

\begin{corollary}
$\vdequiv$ is in $\RP$ when $f$ is given  as a black-box.
\end{corollary}
\begin{proof}
The result immediately follows from Algorithm~\ref{algo-vdproj2} and Theorem~\ref{thm:vdequiv}.\\
\vspace*{5 mm}
  \begin{algorithm}[H]
    \label{algo-vdproj2}
    \SetKwInOut{Input}{Input}
    \SetKwInOut{Output}{Output}
    \Input{$f\in\mathbb{F}[x_1,x_2,\ldots,x_n]$ as a black-box}
    \Output{\texttt{`$f$ is linearly equivalent to $\vd(x_1,\ldots,x_n)$'} if $f\lin {\sf VD}$. Else \texttt{`No such equivalence exists'}}  
     
        Run Kaltofen's factorization Algorithm~\cite{Kal89}\\
        \If{$f$ is irreducible}
        { Output \texttt{`No such equivalence exists'}
      }
    \Else
    { 
    Let $B_1,B_2,\ldots,B_p$ be black-boxes to the irreducible factors of $f$ obtained from Kaltofen's Algorithm.\\
    Interpolate the black-boxes  $B_1,\dots, B_p$ to get the explicit linear forms  $\ell_1,\ell_2,\ldots,\ell_p$ respectively.\\
	Run Algorithm \ref{algo-vdproj}  with $\ell_1\dots \ell_p$ as input.
    }     \caption{$\vdequiv-2$}
\end{algorithm}
\end{proof}
\noindent Finally,   in the black-box setting we show:
\begin{corollary}
$\pit$ is polynomial time equivalent to $ \vdequiv$ in the black-box setting.
\end{corollary}
\begin{proof}
Since polynomial factorization is polynomial time equivalent to $\pit$ in the black-box setting, by Theorem~\ref{thm:vdequiv} we have, $\vdequiv \le_{P} \pit$. For the converse direction, let $f\in\mathbb{F}[x_1,\ldots,x_n]$ be a polynomial of degree $d$. Given black-box access to $f$ we construct black-box to a polynomial $g$ such that $f\equiv 0$ if and only if $g \lin \vd$. Consider the polynomial $g=x_1^{\binom{n}{2}+1}f+ \vd(x_1,x_2,\ldots,x_n)$. 
If $f\equiv 0$ then clearly $g=\vd(x_1,x_2,\ldots,x_n)$. If $f\not\equiv 0$ then $\deg(g) > \binom{n}{2}$ and hence $g$ is not linearly equivalent to $\vd$. %\not\in \vd_{homo}$. 
Observe that given  black-box access to $f$ we can construct in polynomial time a black-box for $g$.
\end{proof}

 \section{Group of symmetries and Lie algebra of Vandermonde determinant}
 In this section we characterize the group of symmetries and Lie algebra of the Vandermonde polynomial. 
 
 \begin{theorem}
 	\label{thm:group-of-sym}
 	Let $\vd$ denote the determinant of the symbolic $n\times n$ Vandermonde matrix.  Then, 
 	\begin{center}
 		$\mathscr{G}_{\vd}= \{ (I + (v\otimes 1) )\cdot P\mid P \in A_n,  v\in\mathbb{F}^n\}$
 	\end{center}
 	where $A_n$ is the alternating group on $n$ elements.
 \end{theorem}
 \begin{proof}
 	%Let $A$ be a matrix in $\mathscr{G}_{\vd}$. 
% 	\begin{enumerate}
 		We first argue the forward direction. Let $A=B + (v\otimes 1)$ where $B\in A_n$ and $v=(v_1,v_2,\ldots,v_n)\in\mathbb{F}^n$. We show that $A\in\mathscr{G}_{\vd}$ :  Let $\sigma$ be the permutation defined by the permutation matrix $B$. Then the transformation defined by $A$ is $A\cdot x_i =  x_{\sigma(i)}+ \sum_{i=1}^{n}v_ix_i$. Now it is easy to observe that $\prod_{i<j}(x_i-x_j) = \prod_{i<j}((A\cdot x_i)-(A\cdot x_j))$. Therefore $A\in\mathscr{G}_{\vd}$.
 	
 	 For the converse direction, consider $A\in\mathscr{G}_{\vd}$. To show that  $A = B + (v\otimes 1)$ where $B\in A_n$ and $v=(v_1,v_2,\ldots,v_n)\in\mathbb{F}^n$.  $A$ defines a linear transformation on the set of variables $\{x_1,x_2,\ldots,x_n\}$ and let  $\ell_i = A\cdot x_i$. We have $\prod_{i<j}(x_i-x_j) = \prod_{i<j}(\ell_i-\ell_j)$. By unique factorization of polynomials, we have that there exists a bijection $\sigma:\{(i,j) \mid i<j\} \rightarrow \{(i,j) \mid i<j\}$ such that  $\sigma(i,j)=(i',j')$ iff $\ell_i-\ell_j=x_{i'}-x_{j'}$.

 		We now show that the $\sigma$ is induced by a permutation $\pi\in S_n$: 
 		\begin{claim}
 		\label{claim:claim1}
 			Let $\sigma$ be as defined above. Then there exists a permutation $\pi$ of $\{1,\ldots,n\}$ such that $\sigma (i,j) = (\pi(i), \pi(j))$.
 					\end{claim}
 	\noindent \textbf{Proof of Claim \ref{claim:claim1}} : Let $G$ be a complete graph on $n$ vertices  such that  edge $(i,j)$ is labelled by $(\ell_i -\ell_j)$ for  $i<j$.  Let $H$ be the complete graph on   $n$ vertices with the  edge $(i,j)$  labelled by $(x_i-x_j)$ for  $i<j$.   Now $\sigma$ can be viewed as a bijection from $E(G)$ to $E(H)$.   It is enough to argue that    for any $1\le i\le n$,
  		 		\begin{align}
 		\nonumber \sigma (\{(1 ,i), (2,i), \ldots, ({i-1}, {i}), ({i},{i+1}), \ldots, (i,n)\}) &= \\
 		  \{(1, {k_i}), (2,k_i), \ldots, (k_i-1, k_i), (k_i, k_i+1), \ldots, (k_i, n)\}
 		 \label{eq:perm}
 		  		\end{align}
 		  		 for some unique $k_i\in[n]$. Then $\pi: i \mapsto k_i$ is the required permutation. 
% 		 The edge $(i',j')\in E(H)  , i'>j'$ iff there exists $(i,j), i>j$ such that $\sigma(i,j)=(i',j')$. We observe the following :
% 				\begin{enumerate}
% 					\item[(i)] $\sigma$ is a bijection between $E(G)$ and $E(H)$.
% 					\item[(ii)] $H$ is a complete graph on vertices $\{1,2,\ldots,n\}$.
% 					\item[(iii)] $\sigma(i,j)=(\pi(i),\pi(j))$ iff $G\simeq  H$.
% 				\end{enumerate}

\noindent  For the sake of contradiction, suppose that (\ref{eq:perm})  is not satisfied for some $i\in[n]$.   Then there are distinct $j,k, m \in [n]$ such that the edges 
$\{ (i,j), (i,k), (i,m)\}$ in $G$ under $\sigma$ map to edges in $\{(\alpha,\beta),(\gamma,\delta),(\eta,\kappa)\}$ in $H$  where the edges $(\alpha,\beta),(\gamma,\delta)$ and $(\eta,\kappa)$ do not form a star in $H$. Note that $\alpha,\beta,\gamma,\delta,\eta,\kappa$ need not be distinct. Various possibilities for the vertices $\alpha,\beta,\gamma,\delta,\eta,\kappa$ and the corresponding vertex-edge incidences in $H$ are depicted in the Figure~\ref{fig:graph1}. Observe that in the figure  the edges are labelled with a $\pm$ sign to denote that based on whether $i<j$ or $j<i$ one of $+$ or $-$ is chosen.

%the image of $K_{1,3}$ $\{ (i,j), (i,k), (i,m)\}$ in $H$ under $\sigma$ is not a $K_{1,3}$, i.e.,   $\sigma(\{(\ell_i -\ell_k), (\ell_i- \ell_j), (\ell_i - \ell_m)\})$ is the set $\{(x_i -x_j), (x_k-x_l), (x_m-x_t)\}$, where the edges $(i,j), (j,l) $ and $(m,t)$ do not form a $K_{1,3}$ in $H$.
  
% 				 {\bf TODO: The following cases need to be cleaned up.}
\begin{figure}[H]
\includegraphics[scale=0.6]{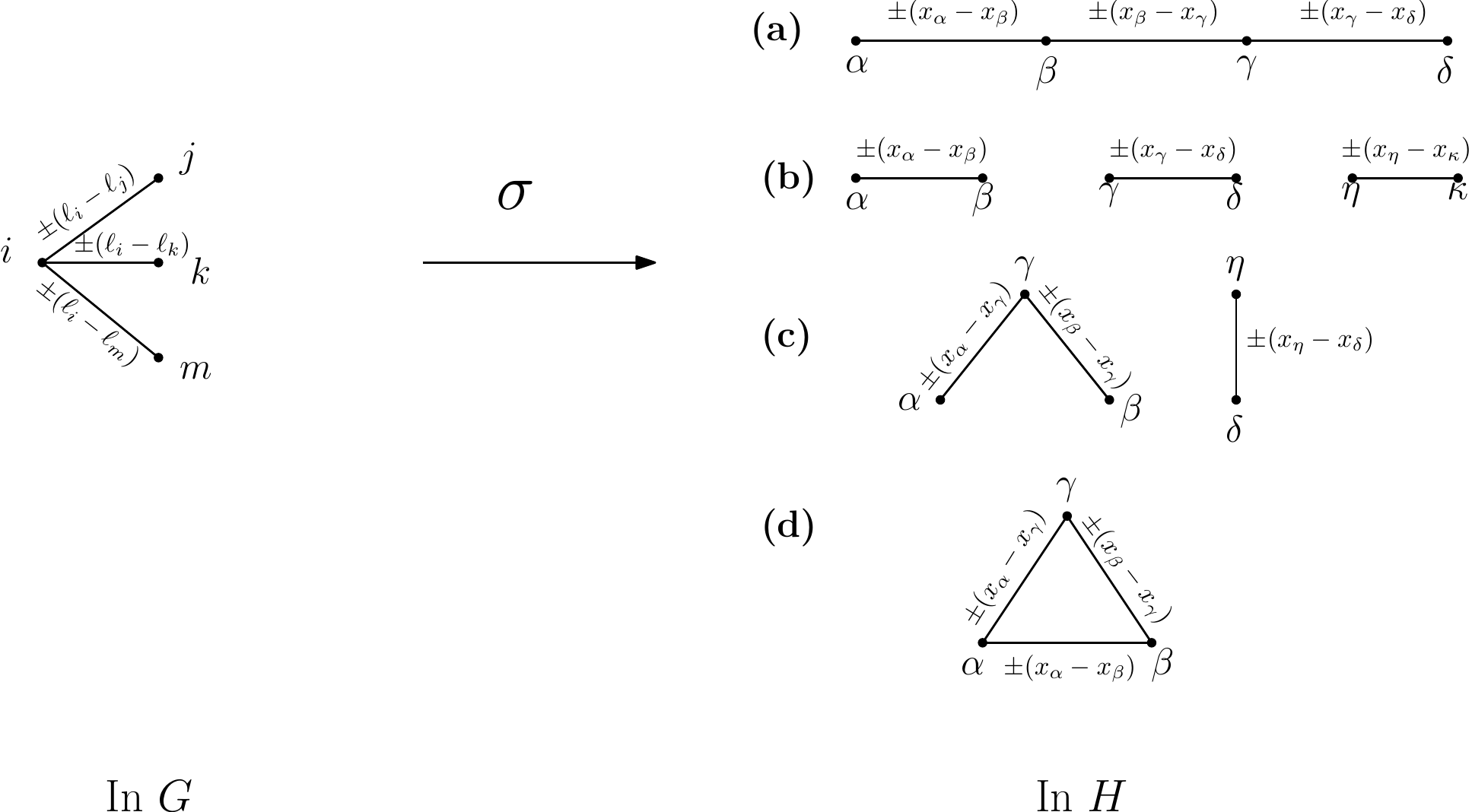}
\caption{The map $\sigma$ on vertex $i$ in $G$}
\label{fig:graph1}
\end{figure}
% 			\begin{center}
% 				
% 			\end{center}
\noindent Recall that we have, 
\begin{equation}
 		\label{eq:var}		\forall i<j~~ |\var(\ell_i-\ell_j)|=2.
 		\end{equation}
We denote by $P$ the edges $\{ (i,j), (i,k), (i,m)\}$ in $G$. Consider the following two cases :
\begin{itemize}
\item[]\textbf{Case 1 :} $P$ in $G$ maps to one of $(a),(b)$ or $(c)$ in $H$ under $\sigma$ (see Figure~\ref{fig:graph1}). In  each of the possibilities,  it can be seen that  there exist linear forms $\ell'$ and $\ell''$ in $\{\ell_i,\ell_j,\ell_k,\ell_m\}$ such that $|\var( \pm(\ell'-\ell''))|=4$ which is a contradiction to Equation \ref{eq:var}.
\item[]\textbf{Case 2 :} $P$ in $G$ maps to  $(d)$ in $H$ under $\sigma$ (see Figure~\ref{fig:graph1}). Without loss of generality  suppose $\sigma(i,j) = (\alpha,\beta)$, $\sigma(i,k) = (\alpha,{\gamma})$ and $\sigma(i,m) = ({\beta},{\gamma})$.  Recall that $\sigma(i,j) = (i',j')$ if and only if $\ell_i -\ell_j = x_{i'}- x_{j'}$. Then we get $\ell_j - \ell_k = x_{\beta}-x_{\gamma}$ by the definition of $\sigma$. Therefore, we have $\sigma(j,k) = (\beta,\gamma)=\sigma(i,m)$ which is a contradiction since $\sigma$ is a bijection. 
\end{itemize}

\noindent  Therefore, for all $1\le i\le n$, Equation \ref{eq:perm} is satisfied and there exists  a permutation $\pi$ such that $\sigma(i,j) = ({\pi(i)},{\pi(j)})$.\hfill $\triangleleft$

%Consider the following four cases :
%\begin{itemize}
%\item[]\textbf{Case 1 :} $P$ in $G$ maps to $(a)$ in $H$ under $\sigma$. Then $\sigma^{-1}(x_i - x_j)-\sigma^{-1}(x_k - x_t)$ has $4$ variables which is a contradiction. % to Equation \ref{eq:var}.
%\item[]\textbf{Case 2 :} $P$ in $G$ maps to $(b)$ in $H$ under $\sigma$. Then $\sigma^{-1}(x_i - x_j)-\sigma^{-1}(x_k - x_t)$ has $4$ variables which is a contradiction. % to Equation \ref{eq:var}.
%\item[]\textbf{Case 3 :} $K_{1,3}$ in $G$ maps to $(c)$ in $H$ under $\sigma$.Then $\sigma^{-1}(x_i -x_k)-\sigma^{-1}(x_m-x_t)$ has $4$ variables which is a contradiction. % to Equation \ref{eq:var}.
%\item[]\textbf{Case 4 :} $K_{1,3}$ in $G$ maps to $(d)$ in $H$ under $\sigma$. Without loss of generality let $\sigma\cdot(\ell_i-\ell_j) = x_i-x_j$ and $\sigma\cdot(\ell_i-\ell_k) = x_i-x_k$.  Then $\ell_k-\ell_m = \ell_j - \ell_i$, a contradiction.
% 			\end{itemize}
\vspace*{4 mm}
 Let $P_{\pi}$ be the permutation matrix corresponding to the permutation $\pi$ obtained from the claim above.  To complete the proof, we need to show that $A=P_{\pi} + (v\otimes 1)$ for $v\in\mathbb{F}^n$. 
 			Let 
 			\begin{center}
 				$\ell_1 = a_{11}x_1+a_{12}x_2 + \cdots +  a_{1n}x_n$\\
 				$\ell_2 = a_{21}x_1+a_{22}x_2 + \cdots +  a_{2n}x_n$\\
 				\hspace*{ 14 mm}\vdots\\
 				$\ell_n = a_{n1}x_1+a_{n2}x_2 + \cdots +  a_{nn}x_n$\\
 			\end{center}
 			Now suppose $\pi$ is the identity permutation, i.e., $\sigma(i,j) = (i,j)$ for all $i<j$, therefore 
 			 	%\begin{center}
 			%$\ell_1-\ell_2=x_{\pi(1)}-x_{\pi(2)}$\\
 			%$\ell_1-\ell_2=x_{\pi(1)}-x_{\pi(2)}$\\
 			%\hspace*{14 mm}\vdots \\
 			%$\ell_1-\ell_p=x_{\pi(1)}-x_{\pi(p)}$\\
 			%\hspace*{14 mm}\vdots \\
 			%$\ell_{p-1}-\ell_p=x_{\pi(p-1)}-x_{\pi(p)}$\\
 			%\end{center}
 			%\begin{center}
 			%a_{11}
 			%\end{center}
 			%
 			%we have $\ell_i-\ell_j = x_{\pi(i)}-x_{\pi(j)}$. Hence we have the equations :
 			%\begin{center}
 			%$a_{i\pi(i)}-a_{j\pi(i)}=1$\\
 			%$a_{i\pi(j)}-a_{j\pi(j)}=-1$
 			%\end{center}
 			%For $k\in[p], k\neq i,j$ we have,
 			%\begin{center}
 			% $a_{i\pi(k)}= a_{j\pi(k)}= a_k$(say).
 			%\end{center}
 			 $\ell_1-\ell_2=x_1-x_2, \ell_1 - \ell_2 = x_1 - x_2, \dots, \ell_1-\ell_n=x_1-x_n$. Now,  we have the following system of linear equations  
 			\begin{center}
 				$a_{11}-a_{21}=1$, 
 				$a_{12}-a_{22}=-1$,
 				$a_{13}-a_{23}=0, a_{14}-a_{24}=0,\ldots,a_{1n}-a_{2n}=0$\\
  $a_{11}-a_{31}=1$, 
 				$a_{12}-a_{32}=0$, 
 				$a_{13}-a_{33}=-1, a_{14}-a_{34}=0,\ldots,a_{1n}-a_{3n}=0$ \\  
 			\hspace*{ 13 mm}\vdots\\
 				$a_{11}-a_{n1}=1$, 
 				$a_{12}-a_{n2}=0$,
 				$a_{13}-a_{n3}=0, a_{14}-a_{n4}=0,\ldots,a_{1n}-a_{nn}=-1.$
 		\end{center}
 		From the equations above, it follows that when  $\pi$ is the identity permutation, $A-  I = v\otimes 1$  for some $v\in \mathbb{F}^n$ where $1$ is the  vector with all entries as $1$. When $\pi$ is not identity,  it follows from the above arguments that $\pi^{-1}A = I+ v\otimes 1$ for some $v\in \mathbb{F}^n$. Since ${\vd}((I+v\otimes 1)X) = \vd(X)$, we conclude that $\pi \in A_n$.
 		\end{proof}

\noindent Now we show that Vandermonde polynomial are characterized by  its group of symmetry $\mathscr{G}_{\vd}$.
 \begin{lemma}
Let $f\in\mathbb{F}[x_1,\ldots,x_n]$ be a homogeneous polynomial of degree $\binom{n}{2}$. If $\mathscr{G}_{f}=\mathscr{G}_{\vd}$ then $f(x_1,\ldots,x_n)=\alpha\cdot \vd(x_1,\ldots,x_n)$ for some $\alpha\in\mathbb{F}$.
\end{lemma}
\begin{proof}
Let $f\in\mathbb{F}[x_1,\ldots,x_n]$. Since $\mathscr{G}_{f}=\mathscr{G}_{\vd} =  \{ (I + (v\otimes 1) )\cdot P\mid P \in A_n,  v\in\mathbb{F}^n\}$, $\mathscr{G}_{f}\cap S_n = A_n$. Hence $f$ is an alternating polynomial. By the fundamental theorem of alternating polynomials~\cite{GZ05,Mat}, there exists a symmetric polynomial $g\in\mathbb{F}[x_1,\ldots,x_n]$ such that $f(x_1,\ldots,x_n)=g(x_1,\ldots,x_n)\cdot \vd(x_1,\ldots,x_n)$. Since $\deg(f)=\binom{n}{2} = \deg(\vd(x_1,\ldots,x_n))$, $g=\alpha$ for some $\alpha\in\mathbb{F}$.
\end{proof}
Using the description of $\mathscr{G}_{\vd}$ above, we now describe the Lie algebra of $\mathscr{G}_{\vd}$. 
% We use a definition of Lie algebra from $\cite{Kayal12}$. 

%\begin{lemma}
%\label{lem:lie1}
%$f((\mathbf{1}_n+\epsilon %A)\mathbf{x})-f(\mathbf{x})=\epsilon\left(\sum_{i,j\in[n]}a_{ij}x_j\frac{\partial f}{\partial x_i}\right)$.
%\end{lemma} 
 
% \begin{proof}
 	
% \end{proof}
 
% \begin{lemma}
 %	\label{lem:vd-partial}
% 	Let $f=\prod\limits_{\substack{i>j\\i,j\in[n]}}(x_i-x_j)$ be the determinant of $n\times n$ Vandermonde matrix. Then,
 %	\begin{center}
% 		$\sum_{i\in[n]}\frac{\partial f}{\partial x_i}\equiv 0.$
% 	\end{center}
 %\end{lemma}
 
 %\begin{proof}
 %	
 %\end{proof}

 \begin{lemma}
 	\label{lem:lie-vand}
We have $\mathfrak{g}_{\vd} =  \{v\otimes 1 \mid v\in\mathbb{F}^n\}$.
 \end{lemma}
 \begin{proof}
  We have
  \begin{eqnarray*}
A\in \mathfrak{g}_{\vd} &\iff&  \prod_{i>j}(x_i - x_j + \epsilon (A(x_i)-A(x_j))) = \prod_{i> j }(x_i - x_j)\\
&\iff & A(x_i) = A(x_j) ~~\forall i \neq j \\
&\iff& A = v \otimes 1 ~\mbox{ for some } v \in \mathbb{F}^n.\hspace*{60mm} \qedhere
	 \end{eqnarray*} 
\end{proof}
 
% \begin{corollary}
% 	As  $\mathfrak{g}_{\vd} =  \{v\otimes 1 \mid v\in\mathbb{F}^n\}$, we have $\{e_1\otimes 1,e_2\otimes 1,\ldots,e_n\otimes 1\}$ is a basis for $\mathfrak{g}_{\vd}$. Hence $dim(\mathfrak{g}_{\vd})=n.$
% \end{corollary}
 
 \begin{defn}{\em (Simple Lie Algebra).}
 	A lie algebra $\mathfrak{g}$ is said to be simple if it is a non-abelian lie algebra whose only ideals are $\{0\}$ and $\mathfrak{g}$ itself.
 \end{defn}

 \begin{corollary}
 	$\mathfrak{g}_{\vd}= \{v\otimes 1 \mid v\in\mathbb{F}^n\}$ is a simple Lie Algebra.
 \end{corollary}
 \begin{proof}
 	Let $\mathfrak{g}=\mathfrak{g}_{\vd}$. Suppose , let $I\subseteq\mathfrak{g}$ such that $I\neq\{0\}$ and $I\neq\mathfrak{g}$. Define the Lie bracket $$[\mathfrak{g},I]=\{[A,B]\mid A\in\mathfrak{g},B\in I\} = \{AB-BA \mid A\in\mathfrak{g},B\in I\}\subseteq I$$ 
 	Observe that $\{e_1\otimes 1,e_2\otimes 1,\ldots,e_n\otimes 1\}$ is a basis for $[\mathfrak{g},I]$. Since $[\mathfrak{g},I] \subseteq I$ we have $n=dim([\mathfrak{g},I])\leq dim(I)$. Also $I\subseteq \mathfrak{g}$ implies that $dim(I)\leq dim(\mathfrak{g}) = n$. Hence $dim(I)=n$. As $I$ is a subspace of the vector space $\mathfrak{g}$, and the $dim(I)=dim(\mathfrak{g})$ we have $I=\mathfrak{g}$. 
 	 \end{proof}

\section{Models of Computation}
\label{sec:models}
In this section we study polynomials that can be represented as projections of Vandermonde polynomials. Recall the definitions of the classes  $\vd,\vd_{\proj},\vd_{\homo}$ and $\vd_{\aff}$ from Section \ref{sec:prelim}. For any arithmetic model of computation, universality and closure under addition and multiplication are among the most fundamental and necessary properties to be investigated. Here, we study these properties for projections of the Vandermonde polynomial and their sums. Most of the proofs follow from elementary arguments and can be found in the Appendix.

%\begin{eqnarray*}
%\vd =\{f \mid f=\prod\limits_{\substack{i>j\\i,j\in[n]}}(x_i-x_j)\};~~~
%\vd_{\proj} = \{ f \mid f\leq_{\proj}g , g\in\vd \} \\
%\vd_{\homo} = \{ f \mid f\leq_{\homo}g , g\in\vd \};~~~
%\vd_{\aff} = \{ f \mid f\leq_{\aff}g , g\in\vd \}. 
%\end{eqnarray*}
By definition,  $\vd,\vd_{\proj},\vd_{\homo} \subseteq \vd_{\aff}$. Also, any polynomial with at least one irreducible non-linear factor cannot be written as a projection of $\vd$. As expected, we observe that there are products of linear forms that cannot cannot be written as a projection of $\vd$.

\begin{lemma}
\label{lem:vd-univ}
Let $(x_1-y_1)(x_2-y_2)\not\in \vd_{\aff}$.
\end{lemma}
\begin{proof}
Suppose $f\in\vd_{\aff}$, then there  are affine linear forms $\ell_1,\ldots, \ell_n$ such that $(x_1-y_1)(x_1-y_2) = \prod_{\substack{1\le i< j\le n}}(\ell_i-\ell_j)$. Clearly, only two factors of $\prod_{1 \leq i< j \leq n}(\ell_i-\ell_j) $ are non constant polynomials.
Without loss of generality, let $\ell_i-\ell_j=x_1-y_1$ and $\ell_{i'}-\ell_{j'}=x_2-y_2$. Then,  we must have $\ell_{i'}-\ell_i$, $\ell_{j'}-\ell_j$, $\ell_{i'}-\ell_{j}$ and $\ell_{j'}-\ell_i$  as constant polynomials, as they are factors of $\vd(\ell_1,\ldots ,\ell_n)$ and hence $\ell_{i'}-\ell_{j'} = \ell_{i'}-\ell_i -(\ell_{j'} - \ell_i)$ is a constant, which is a contradiction. 
\end{proof}

\begin{lemma}
\label{lem:closure}
The classes $\vd, \vd_{\proj}, \vd_{\hom}$ and $\vd_{\aff}$ are  not closed under addition and multiplication. 
\end{lemma}
\begin{proof}
Since sum of any two variable disjoint polynomials is irreducible, it is clear that $\vd, \vd_{\proj}, \vd_{\hom}$ and $\vd_{\aff}$ are  not closed under addition.  For multiplication, take $f_1 = x_1-y_1$, and $f_1=x_2-y_2$. By Lemma~\ref{lem:vd-univ}, $f_1f_2 \notin \vd_{\aff}$ and hence $f_1f_2\notin \vd \cup \vd_{\proj}\cup \vd_{\hom}$. Since $f_1,f_2 \in \vd \cap \vd_{\proj}\cap \vd_{\hom}\cap \vd_{\aff}$, we have that $\vd, \vd_{\proj}, \vd_{\hom}$ and $\vd_{\aff}$  are not closed under multiplication.
\end{proof}

It can also be seen that the classes of polynomials $\vd, \vd_{\proj}, \vd_{\hom}$ and $\vd_{\aff}$ are properly separated from each other:

\begin{lemma}
\label{lem:sep}
\begin{enumerate}
\item[(1)] $\vd_{\proj}\subsetneq \vd_{\aff}$ and $\vd_{\homo}\subsetneq \vd_{\aff}$. 
\item[(2)]   $\vd_{\proj} \not\subset \vd_{\homo}$ and $\vd_{\homo} \not\subset \vd_{\proj}$.   
\end{enumerate}
\end{lemma}

\begin{proof}
\begin{itemize}
\item $\vd_{\proj}\subsetneq \vd_{\aff}$ : Let $f=(x_1-y_1)+(x_2-y_2)$.Then  $f = \det \left[ {\begin{array}{cc}
1 & 1 \\
y_2-x_2 & x_1-y_1  
\end{array} } \right]
$.   By comparing factors it can be seen that $(x_1-y_1)+(x_2-y_2)\not\in\vd_{\proj}$. 
\item $\vd_{\homo}\subsetneq \vd_{\aff}$ : Let $f=x_1+x_2-2$.  Then  $f=\det \left[ {\begin{array}{cc}
1 & 1 \\
x_2-1 & x_1-1  
\end{array} } \right]
$. Suppose $f\in\vd_{\homo}$, then 
there exists an $n\times n$ Vandermonde matrix $M'$ such that $f\leq_{\homo} \det(M')$. In other words, 
$x_1+x_2-2 = \prod_{\substack{i< j~~i,j\in[n]}}(\ell_j-\ell_i)$.
where $\ell_i$'s are homogeneous linear forms which is impossible since $x_1+x_2-2$ is non-homogeneous.

\item $\vd_{\homo}\subsetneq\vd_{\proj}$ : Let $f=(x_1-1)(x_1-2)$. Observe that $f\in\vd_{\proj}$.  However, since $\vd_{\homo}$ consists only of polynomials with homogeneous linear factors,  $f\notin\vd_{\homo}$.

\item $\vd_{\proj}\subsetneq\vd_{\homo}$ : Let $f=(x_1-y_1)+(x_2-y_2)$. For $M=\left[ {\begin{array}{cc}
1 & 1 \\
y_2-x_2 & x_1-y_1  
\end{array} } \right]
$, we have $\det(M)\in\vd_{homo}$ and $f=\det(M)$. It can be seen that $(x_1-y_1)+(x_2-y_2)\not\in\vd_{\proj}$. \qedhere  
\end{itemize}
\end{proof}

\subsection*{Sum of projections of Vandermonde polynomials}
In this section, we consider polynomials that can be expressed as sum of projections of Vandermonde polynomials. 
  
\begin{defn} For a class $\mathcal{C}$ of polynomials, let $\Sigma\cdot{\cal C}$  be  defined as
$$\Sigma\cdot{\cal C} = \left \{f ~\vline~ \parbox{4.3in}{$f=(f_n)_{n\ge 0}$ where $\forall n\ge 0~~\exists~g_1,g_2,\ldots,g_s\in{\cal C}$,$\alpha_1,\dots,\alpha_s\in\mathbb{F}$ such that $f~=~\alpha_1 g_1+\alpha_2g_2+\cdots+\alpha_sg_s, s = n^{O(1)}$}\right\}.$$
\end{defn}

%\exists~g_1,g_2,\ldots,g_s\in{\cal C}, \alpha_1,\dots,\alpha_s\in\mathbb{F} \text{~such that~}  f~=~\alpha_1 g_1+\alpha_2g_2+\cdots+\alpha_sg_s, s = n^{O(1)}$

\begin{lemma}
\label{lem:sigvd:univ}
$x_1\cdot x_2 \not\in \Sigma\cdot\vd$.
\end{lemma}
\begin{proof}
Suppose   there exists $g_1,g_2,\ldots,g_s\in\vd$.  Note that  for every $i$, either $\deg(g_i) \le 1$ or $\deg(g_i) \ge 3$.  Since $\deg(g) = 2$, it is impossible that $x_1x_2 = g_1+ \dots + g_s$ for any $s\ge 0$.  
\end{proof}

\begin{lemma}
The class $\Sigma\cdot\vd$ is closed under addition but not under multiplication.
\begin{itemize}
\item[(i)] If $f_1,f_2\in\Sigma\cdot\vd$ then $f_1+f_2\in\Sigma\cdot\vd$.
\item[(ii)] There exists $f_1,f_2\in\Sigma\cdot\vd$ such that $f_1\cdot f_2 \not\in \Sigma\cdot\vd$
\end{itemize}
\end{lemma}
\begin{proof}
\begin{itemize}
\item[(i)] Closure under addition follows by definition.
\item[(ii)] Let $f_1=x_1-y_1$ and $f_2 = x_2 -y_2$, clearly $f_1,f_2\in\Sigma\cdot\vd$. since for any $g \in \vd$, $\deg(g) \neq 2$, one can conclude that $f_1f_2 \notin \Sigma\cdot \vd$.    \qedhere \end{itemize}\end{proof}

 We now  consider polynomials in the class $\Sigma\vd_{\proj}$. Any univariate polynomial $f$ of degree $d$ can be computed by depth-$2$ circuits of size $\poly(d)$. However there are univariate polynomials not in $\vd_{\aff}$ which is a subclass of depth $2$ circuits (Consider any univariate polynomial irreducible over $\mathbb{F}$). Here, we show that the class of all univariate polynomials can be computed efficiently by  circuits in $\Sigma\vd_{\proj}$.

\begin{lemma}
\label{lem:univariate}
Let $f=a_0+a_1x+ a_2x^2+ \cdots + a_dx^d$ be a univariate polynomial of degree $d$. Then  there are    $g_i\in\vd_{\proj}$, $1\le i \le s \le O(d^2)$   for some $ \alpha_i\in\mathbb{F}$ such that $f = g_1 + \dots + g_s$.
\end{lemma}
\begin{proof}
Consider the $(d+1)\times (d+1)$ Vandermonde matrix $M_0$,
\begin{equation}
\nonumber M_0 =  \left[ {\begin{array}{ccccc} 
 1 & 1 & 1 & 1 & 1 \\
 x & \beta_1 & \cdots & \cdots & \beta_{d-1}\\
x^2 & \beta_1^2 & \cdots &\cdots & \beta_{d-1}^2\\
\vdots & \vdots & \vdots & \vdots & \vdots \\
\vdots & \vdots & \vdots & \vdots & \vdots \\
x^{d-1} & \beta_1^{d-1} & \cdots &\cdots & \beta_{d-1}^{d-1}\\
x^{d} & \beta_1^{d} & \cdots &\cdots & \beta_{d-1}^{d}
\end{array} } \right]
\end{equation}

%\end{equation}
Let   $g_0 = \det(M_0) = \gamma_{00} + \gamma_{01}x + \gamma_{02}x^2 + \cdots + \gamma_{0,d-1}x^{d-1} + \gamma_{0d}x^d$ where $\gamma_{00},\ldots,\gamma_{0d}\in\mathbb{F}$ and $\gamma_{0d} \neq 0$. Note that $g_0\in \vd_{\proj}$.   Setting  $\alpha_0 = \frac{a_d}{\gamma_{0d}}$ we get    $\alpha_0f_0 = a_dx^d + \frac{a_d\gamma_{0,d-1}}{\gamma_{0d}}x^{d-1} + \cdots + \frac{a_d\gamma_{01}}{\gamma_{0d}}x + \frac{a_d\gamma_{00}}{\gamma_{0d}}.$  Now,  let $M_1$ be the  $d\times d$ Vandermonde matrix, 
\begin{equation}
\nonumber M_1 =  \left[ {\begin{array}{ccccc} 
 1 & 1 & 1 & 1 & 1 \\
 x & \beta_1 & \cdots & \cdots & \beta_{d-1}\\
x^2 & \beta_1^2 & \cdots &\cdots & \beta_{d-1}^2\\
\vdots & \vdots & \vdots & \vdots & \vdots \\
\vdots & \vdots & \vdots & \vdots & \vdots \\
x^{d-1} & \beta_1^{d-1} & \cdots &\cdots & \beta_{d-1}^{d-1}\\
\end{array} } \right]
\end{equation}
Then 
$$g_1 = \det(M_1) = \gamma_{10} + \gamma_{11}x + \gamma_{12}x^2 + \cdots + \gamma_{1,d-1}x^{d-1}.$$ where $\gamma_{10},\ldots,\gamma_{1d}\in\mathbb{F}$. Observe that $x^d$ is not a monomial in $\alpha_1g_1$. Set $\alpha_1 = \frac{a_{d-1}}{\gamma_{1,d-1}}- \frac{a_d\gamma_{0,d-1}}{\gamma_{0d}} $. Then $\alpha_1g_1 = a_{d-1}x^{d-1} + \cdots + \frac{a_d\gamma_{01}}{\gamma_{0d}}x + \frac{a_d\gamma_{00}}{\gamma_{0d}}.$
Extending this approach : Let $M_i$ be a $(d-(i-1))\times(d-(i-1))$ Vandermonde matrix,
%\begin{equation}
%\nonumber
%M_i =  \left[ {\begin{array}{cccccc}
%1 & x_j & \beta_1  & \cdots & \cdots & \beta_{d-(i+2)}\\
%1 & x_j^2 & \beta_1^2  & \cdots & \cdots & \beta_{d-(i+2)}^2\\
%\vdots & \vdots & \vdots & \vdots & \vdots \\
%\vdots & \vdots & \vdots & \vdots & \vdots \\
%1 & x_j^{d-i} & \beta_1^{d-i} & \cdots & \cdots & \beta_{d-(i+2)}^{d-i}\\
%\end{array} } \right]
%\end{equation} 

\begin{equation}
\nonumber M_i =  \left[ {\begin{array}{ccccc} 
 1 & 1 & 1 & 1 & 1 \\
 x & \beta_1 & \cdots & \cdots & \beta_{d-i}\\
x^2 & \beta_1^2 & \cdots &\cdots & \beta_{d-i}^2\\
\vdots & \vdots & \vdots & \vdots & \vdots \\
\vdots & \vdots & \vdots & \vdots & \vdots \\
x^{d-i} & \beta_1^{d-i} & \cdots &\cdots & \beta_{d-i}^{d-i}\\
\end{array} } \right]
\end{equation}
Now observe that by setting $\alpha_i = \frac{a_{d-i}}{\gamma_{i,d-i}}- (\alpha_0\gamma_{0,d-i} + \alpha_1\gamma_{1,d-i} + \cdots + \alpha_{i-1}\gamma_{i-1,d-i}) $ we ensure that
$\sum_{j=0}^{i}\alpha_jg_j$ does not contain any term of the form $x^p$ for $d-i \leq p \leq d-1$. Thus  $\sum_{k=0}^{d}\alpha_kg_k = a_dx^d$. Hence to compute  $a_dx^d$ we require $d$ summands. Then, using  $O(d^2)$ summands  $f$ can be obtained.
\end{proof}

%\noindent Given that $\Sigma\cdot\vd_{\proj}$ can compute univariate polynomials we ask what classes of multivariate polynomials can be computed by the model in polynomial size.
%We show that quadratic polynomials can be computed in polynomial size.  

%
%\begin{defn}{\em (Sylvester law of inertia for quadratic forms).}
%\label{def:syl}
%A Quadratic form $Q(x_1,\ldots,x_n)$ is a homogeneous polynomial of degree $2$ in $n$ variables$(a_{11}x_1^2+a_{22}x_2^2 + \cdots+a_{nn}x_n^2 + \sum_{i<j}a_{ij}x_ix_j$ 
%where $a_{ij}\in\mathbb{F},i,j\in[n])$.
%Sylvester law of inertia states that any for quadratic form $Q(x_1,\ldots,x_n)$ there exists a linear transformation 
%\begin{center}
%$y_1=\gamma_{11}x_1 + \cdots + \gamma_{1n}x_n$\\
%$y_2=\gamma_{21}x_1 + \cdots + \gamma_{2n}x_n$\\
%\hspace{8 mm}\vdots\\
%$y_n=\gamma_{n1}x_1 + \cdots + \gamma_{nn}x_n$
%\end{center}
%and $b_1,b_2,\ldots,b_n\in\{0,1,-1\}$ such that $Q(x_1,x_2,\ldots,x_n)= \sum_{i=1}^nb_iy_i^2$.
%\end{defn}
%\noindent By the Sylvester law of inertia for quadratic forms in Definition \ref{def:syl}, it suffices to show that for any $i\in[n]$, $x_i^2\in\vd_{\proj}$ which is true by the proof of Lemma \ref{lem:pow-sym}. Hence we have the following corollary.
%
%\begin{corollary}
%Let $Q(x_1,\ldots,x_n)$ be a quadratic form in $\mathbb{F}[x_1,\ldots,x_n]$.
%If $Q(x_1,\ldots,x_n)=\sum_{i=1}^s\alpha_if_i$ where $f_i\in\vd_{\proj}$ then $s\leq 3n$.
%\end{corollary}

Recall that the $n$-variate power symmetric polynomial of degree $d$ is defined as ${\sf Pow_{n,d}} = x_1^d+x_2^d+\cdots + x_n^d.$
From the arguments in  Lemma~\ref{lem:univariate},  it follows that ${\sf Pow_{n,d}}$ can be expressed by polynomial size circuits in $\Sigma\cdot\vd_{\proj}$.

\begin{corollary}
\label{cor:pow-sym}
There are polynomials $f_i\in\vd_{\proj}$ $1 \le i \le nd$ such that 
${\sf Pow_{n,d}}= \sum_{i=1}^s\alpha_if_i$. 
\end{corollary}

Now, to argue that $\vd_{\homo}$ and $\vd_{\aff}$ are universal, we need the following:
\begin{lemma}[\cite{Fis94}]
\label{lem:lem1}
Over any infinite field containing the set of integers, there exists $2^d$ homogeneous linear forms $L_1,\ldots,L_{2^d}$ such that
\begin{center}
$\prod\limits_{i=1}^d x_i = \sum\limits_{i=1}^{2^d} L_i^d$
\end{center} 

\end{lemma}
Combining with Corollary~\ref{cor:pow-sym} with Lemma~\ref{lem:lem1} we establish the universality of the classes $\Sigma\cdot\vd_{\homo}$ and $\Sigma\cdot\vd_{\aff}$.

\begin{lemma}
The classes $\Sigma\cdot\vd_{\homo}$ and  $\Sigma\cdot\vd_{\aff}$ are universal.
\end{lemma}

Also, in the following, we note that  $\vd_{\aff}$ is more powerful than depth three $\Sigma\wedge\Sigma$ circuits: 

\begin{lemma}${\sf \poly-size~}\Sigma\wedge\Sigma\subsetneq {\sf \poly{-}size~}\Sigma\cdot \vd_{\aff}$.
\end{lemma}
\begin{proof}
Let $f\in\Sigma\wedge\Sigma$. Then $f=\sum\limits_{i=1}^{s}\ell_i^d$. Then for any $k\ge 1$, $\dim \partial^{=k} (f) \le s$.  Now, for any $ 1\le k\le n/2$  we have  $\dim \partial^{=k} (\vd) \ge {n \choose k}$. Therefore,  if $f= \vd$ we have $s = 2^{\Omega(n)}$  by setting $k = n/2$. Hence ${\sf \poly-size~}\Sigma\wedge\Sigma\subsetneq {\sf \poly{-}size~}\Sigma\cdot \vd_{\aff}$.   
%\textcolor{red}{dimension of shifted partials}
\end{proof}

%
%\begin{center}
%\begin{table}
%\begin{tabular}{| c | c  | c  | c | c | c |} \hline
% Model & Universality & Expressive Power & Closure under + & Closure under $\times$ & PIT \\ \hline
%\vd & No & $\prod\limits_{1 \leq i<j\leq n}(x_j-x_i)$ &  No & No \\ \hline
%$\vproj$ & No & & No & No \\ \hline
%$\vaff$ & No & $\prod\limits_{1 \leq i<j\leq n}(l_j-l_i)$ & No & No \\ \hline
%$\Sigma\vd$ & No & $\sum\limits_{i=1}^s\prod\limits_{1 \leq i<j\leq n}(l_j-l_i)$ & Yes & No \\ \hline
%$\Sigma\vproj$ &  & & &\\ \hline
%$\Sigma\vaff$ & Yes & \begin{tabular}{@{}c@{}}contains $\powsym$ \\ contains Univariate polynomials\end{tabular} & Yes & Yes \\ \hline
%$\Sigma\Pi\vd$ & Yes & \begin{tabular}{@{}c@{}}contains $\powsym$ \\ contains $\Sigma\Pi$ circuits \end{tabular} & Yes & Yes \\ \hline
%$\Sigma\Pi\vproj$ & & & \\\hline
%$\Sigma\Pi\vaff$ & Yes & \begin{tabular}{@{}c@{}c@{}}contains $\powsym$ \\ contains $\Sigma\Pi\Sigma$ circuits \\ contains Univariate polynomials \end{tabular} & Yes & Yes \\ \hline
%\end{tabular}
%\caption{Models under consideration}
%\end{table}
%\end{center}

\subsection*{A Lower Bound against $\Sigma\cdot \vd_{\proj}$}
%\begin{defn}
%Let $f\in\mathbb{F}[x_1,\ldots,x_n]$ be a polynomial. The vanishing subspace $V_f$ of $f$ is defined as the affine space such that    $f\vert_{V_f}\equiv 0$. 
%\end{defn}
%
%
%\noindent We prove a linear lower bound for $\Sigma\cdot\vd_{\aff}$ computing $\perm_n$ using 
%co-dimension of the vanishing subspace. 
%
%\begin{lemma}
%If $\sum\limits_{i=1}^s\alpha_if_i=\det_n$ where $f_i\in\vd_{\aff}$ then $s\geq n-1$.
%\end{lemma}
%\begin{proof}
%Consider the polynomial $f$ in $\Sigma\cdot\vd_{\aff}$. Then $f=\sum\prod\limits_{\substack{i,j\in[n]\\i>j}}(\ell_i-\ell_j)$. Let $V_f$ be the subspace such that $f\vert_{V_f}\equiv 0$ Observe that the solution space to the polynomial equation $\ell_i-\ell_j=0$ for any $i>j,i,j\in[n]$ is the vanishing subspace of $f$. Hence $\codim(V_f)=1$. Also observe that $\codim(\det_n)=n-1$. Now, 
%Let $\alpha_1f_1+\alpha_2f_2 + \cdots + \alpha_sf_s = \det_n$ where each $f_i\in\vd_{\aff}$.   The vanishing subspace of $\alpha_1f_1+\alpha_2f_2 + \cdots + \alpha_sf_s$ is $V_1\cup V_2 \cup \cdots \cup V_s$ where each $V_i$ is the vanishing subspace of $f_i$. Therefore $\codim(\det_n)=\codim(V_1\cup V_2 \cup \cdots \cup V_s)$. Hence $s\geq n-1$.
%\end{proof}

 Observe that $\Sigma\cdot\vd_{\proj}$ is a subclass of non-homogeneous depth circuits of bottom fan-in $2$, i.e., $\Sigma\Pi\Sigma^{[2]}$. It is known that $Sym_{2n,n}$ can be computed by non-homogeneous $\Sigma\Pi\Sigma^{[2]}$ circuits of size $O(n^2)$. We show that any $\Sigma\cdot \vd_{\proj}$ computing $Sym_{n,n/2}$ requires  a top fan-in of $2^{\Omega(n)}$ and hence $\Sigma\cdot\vd_{\proj}\subsetneq \Sigma\Pi\Sigma^{[2]}$.   The lower bound is obtained by   a variant of the evaluation dimension as a complexity measure for polynomials. % - \textit{restricted evaluation dimension}.

\begin{defn}{\em (Restricted Evaluation Dimension.)}
Let $f\in\mathbb{F}[x_1,\ldots,x_n]$ and $S=\{i_1,\ldots,i_k\}\subseteq [n]$. Let $\bar{a}=(a_{i_1},a_{i_2},\ldots,a_{i_k})\in\{0,1,*\}^k$
and $f\vert_{S=\bar{a}}$ be the polynomial obtained by substituting for all $i_j\in S$,
\begin{center}
$x_{i_j} = \begin{cases} 1 &\mbox{if } a_{i_j}=1 \\
0 &\mbox{if } a_{i_j}=0 \\
x_{i_j} &\mbox{if } a_{i_j}=* \\
\end{cases} $
\end{center}
Let  $f\vert_{S} \stackrel{\sf def}{=} \{ f\vert_{S=\bar{a}} \mid \bar{a}\in\{0,1,*\}^k \}$. The {\em restricted evaluation dimension of $f$}  is defined as: 
\begin{center}
$\ed(f) \stackrel{\sf def}{=} dim(\mathbb{F}\mbox{-}span(f\vert_{S}))$ 
\end{center}
\end{defn}
It is not hard to see that the measure $\ed$ is sub-additive: % and sub-mutliplicative:    
\begin{lemma}
\label{lem:red-subadd}
%The measure Restricted Evaluation Dimension is sub-additive and sub-multiplicative.
 For any $f,g\in\mathbb{F}[x_1,x_2,\ldots,x_n]$,
$\ed(f+g)\leq \ed(f) + \ed(g)$.
\end{lemma}  
In the following, we show that Vandermonde polynomials and their projections have low restricted evaluation dimension:

\begin{lemma}  
\label{lem:vd-ub}
Let $M$ be a $m\times m$ Vandermonde matrix with entries from $\{x_1,\dots, x_n\}\cup \mathbb{F}$ and $f=\det(M)$. %$\{x_{i}^{e_i}~|~1\le i\le n, e_i \in \mathbb{N}\}\cup\mathbb{F}$, $m\leq n$ and $f=\det(M)$. 
Then for any $S\subset \{1\ldots, n\}$ with $|S|=k$, we have  $\ed(f)\leq (k+1)^2$.
\end{lemma}
\begin{proof}
Without loss of generality suppose  $S = \{j_1,j_2,\ldots,j_k\}\subseteq [n]$ and $|S|=k$. Let $T = \{x_{j_1},x_{j_2},\ldots,x_{j_k}\}\cap \var(f) = \{i_1,i_2,\ldots,i_m\}$. Observe that $m\leq k$. 
For a vector $v\in \{0,1,*\}^n$ and $b\in \{0,1\}$, let $\#_b(v)$ denote the number of occurrences  of $b$ in the vector $v$. Then, for any $\bar{a}=(a_{j_1},a_{j_2},\ldots,a_{j_k})\in\{0,1,*\}^k$,
\begin{itemize}
\item If $\#_{0}((a_{i_1},a_{i_2},\ldots,a_{i_m}))\geq 2$ or $\#_{1}((a_{i_1},a_{i_2},\ldots,a_{i_m}))\geq 2$ then $f\vert_{S=\bar{a}}=0$. 
\item If $\#_{0}((a_{i_1},a_{i_2},\ldots,a_{i_m})) = \#_{1}((a_{i_1},a_{i_2},\ldots,a_{i_m})) = 1$.   Let $T_1$ be the set of polynomials obtained from such evaluations of $f^d$.  The number of such assignments is at most $\binom{m}{2} \leq \binom{k}{2}\leq k^2$ and hence  $|T_1| \le k^2$.
\item If  $\#_{0}(\{a_{i_1},a_{i_2},\ldots,a_{i_m}\})=1$ or $\#_{1}(\{a_{i_1},a_{i_2},\ldots,a_{i_m}\})=1$.  Let $T_2$ denote the set of polynomials obtained from such evaluations. Since number of such assignments is    $2\binom{m}{m-1} \leq 2\binom{m}{1} \leq 2\binom{k}{1} \leq 2k $, we have $|T_2| \le 2k$. 
\item If, $\#_{0}(\{a_{i_1},a_{i_2},\ldots,a_{i_m}\}) =  \#_{1}(\{a_{i_1},a_{i_2},\ldots,a_{i_m}\})= 0$, in this case, the polynomial $f$ does not change under these evaluations.
\end{itemize}
From the above case analysis, we have $\mathbb{F}\mbox{-}span(f|S=a) =\mathbb{F}\mbox{-}span(T_1\cup T_2\cup \{f\})$. Therefore $\ed(f) \leq k^2 + 2k +1 \leq (k+1)^2$.
\end{proof}

\begin{lemma}
\label{lem:sym-lb}
Let $Sym_{n,k}$ be the elementary symmetric polynomial in $n$ variables of degree $k$. Then for any $S\subset \{1,\ldots, n\}$, $|S|=k$, we have $\ed(Sym_{n,k})\geq 2^k-1$.
\end{lemma}
\begin{proof}
Let $Sym_{n,k}$ be the elementary symmetric polynomial in $n$ variables of degree $k$.
For $T\subseteq S, T\neq \emptyset$,  define $\bar{a}_T = (a_1,\ldots,a_k)\in \{1,*\}^k$  as: 
\begin{center}
$a_i = \begin{cases} * &\mbox{if } x_i\in T \\
1 &\mbox{if } x_i \in S\setminus  T \end{cases}$
\end{center}
 Note that it is enough to prove:
\begin{equation}
\label{eq:dim}
\hspace*{30 mm}dim(\{Sym_{n,k}\vert_{S=\bar{a}_T} \mid T\subseteq S, T\neq \emptyset\})\geq 2^k-1
\end{equation}
Since $\{Sym_{n,k}\vert_{S=\bar{a}_T} \mid T\subseteq S, T\neq \emptyset\} \subseteq \{\mathbb{F}\mbox{-}span(Sym_{n,k}\vert_{S=\bar{a}}) \mid \bar{a} = (a_1,\ldots,a_k)\in \{1,*\}^k \}$, by Equation (\ref{eq:dim}) we have
$$\ed(Sym_{n,k}) \geq dim(\{Sym_{n,k}\vert_{S=\bar{a}_T} \mid T\subseteq S, T\neq \emptyset\}) =2^k-1.$$
 To prove~(\ref{eq:dim}), note that for any distinct $T_1,T_2\subseteq S$, we have $Sym_{n,k}\vert_{S=\bar{a}_{T_1}}$ and $Sym_{n,k}\vert_{S=\bar{a}_{T_2}}$ have distinct leading monomials with respect to the {\sf lex} ordering since they have distinct supports.  Since the number of distinct leading monomials in a space of  polynomials is a lower bound on its dimension, this concludes the proof of (\ref{eq:dim}).
\end{proof}

\begin{theorem}
\label{thm:sym-vd-lb}
If $\sum\limits_{i=1}^s\alpha_if_i=Sym_{n,n/2}$ where $f_i\in\vd_{\proj}$ then $s=2^{\Omega(n)}$.
\end{theorem}
\begin{proof}
The proof is a straightforward application of sub-additivity of $\ed$  combined with Lemmas~\ref{lem:sym-lb} and \ref{lem:vd-ub}.
\end{proof}

%\section{Connection to depth-3 circuits}
%\section{Conclusions}

\bibliographystyle{plain} 
\bibliography{refbib} 

\end{document}